\documentclass[a4paper,11pt,english]{amsart}[2004/08/06]
\usepackage{microtype}
\usepackage{fixltx2e}[2006/03/24]
\usepackage[utf8]{inputenc}
\usepackage{svninfo}
\svnInfo$Id$
\usepackage[T1]{fontenc}
\usepackage{lmodern}

\usepackage{natbib}

\usepackage{amssymb}
\usepackage{amsmath,amscd}
\usepackage{mathtools}
\mathtoolsset{mathic, showonlyrefs}
\usepackage{enumitem}
\setenumerate[1]{label=(\arabic*)}
\setenumerate[2]{label=(\alph*), ref=(\theenumi.\alph*)}

\usepackage[english,bookmarksnumbered,bookmarksopen,unicode,pdfusetitle, breaklinks=false,pdfborder={0 0 0},backref=false,colorlinks=false]{hyperref}

\hypersetup{
  pdfauthor={Gábor Braun and Sebastian Pokutta},
  pdfkeywords={},
  pdfsubject={MSC2000 Primary 52B15;
    Secondary 52B05, 20B30}} 

\numberwithin{equation}{section}

\newcommand{\mynewtheorem}[2]{%
  \newtheorem{#1}{#2}[section]%
  \expandafter\DeclareRobustCommand\expandafter{\csname#1autorefname\endcsname}{#2}%
  \expandafter\let
    \csname c@#1\expandafter\endcsname
    \expandafter=%
    \csname c@mytheorem\endcsname}
  {NOT TO USE}[section]

\theoremstyle{plain}
\mynewtheorem{theorem}{Theorem}
\mynewtheorem{lemma}{Lemma}
\mynewtheorem{proposition}{Proposition}
\mynewtheorem{corollary}{Corollary}
\mynewtheorem{conjecture}{Conjecture}
\mynewtheorem{observation}{Observation}
\mynewtheorem{note}{Note}
\mynewtheorem{todo}{*Todo*}

\newtheorem*{ScalarLemma}{Lemma~\ref{lem:invariantScalarAndSection}}

\theoremstyle{definition}
\mynewtheorem{definition}{Definition}
\mynewtheorem{distribution}{Distribution}
\mynewtheorem{example}{Example}
\mynewtheorem{remark}{Remark}
\mynewtheorem{question}{Question}
\mynewtheorem{answer}{Answer}

\newcommand*{\size}[1]{\left|#1\right|}

\newcommand*{\Stab}[2]{{#1}\sb{#2}}
\newcommand*{\orb}[2]{{#1}{#2}}
\newcommand*{\fixedPoints}[2]{{#2}\sp{#1}}

\newcommand{\R}{\mathbb{R}}
\newcommand{\N}{\mathbb{N}}

\DeclareMathOperator{\conv}{conv}
\newcommand{\cube}[1]{[0,1]^{#1}}
\newcommand{\binSet}{\{0,1\}}
\newcommand{\face}[1]{\left\{#1\right\}}
\newcommand{\set}[2]{\left\{#1\,\middle|\,#2\right\}}
\newcommand{\sprod}[1]{\left\langle #1\right\rangle}
\newcommand{\ssprod}[1]{\overline{\langle#1\rangle}}
\newcommand{\tG}{\widetilde{G}}
\newcommand{\tg}{\widetilde{g}}
\newcommand{\onorm}[1]{\left\|#1\right\|_1}
\newcommand{\card}[1]{\left\lvert#1\right\rvert}
\DeclareMathOperator{\supp}{supp}
\DeclareMathOperator{\vertex}{vertex}
\newcommand{\pmatch}{P_{\operatorname{match}}}
\newcommand{\sg}{S}
\newcommand{\ag}{A}
\newcommand{\pA}{\mathcal A}
\newcommand{\pB}{\mathcal B}
\DeclareMathOperator{\xc}{xc}
\newcommand*{\xcs}[1][\ag_{n}]{\operatorname{xc}_{#1}}
\DeclareMathOperator{\relint}{rel.int}

\begin{document}
\title[An algebraic approach to symmetric extended formulations]{An
  algebraic approach to \\ symmetric extended formulations}

\author{Gábor Braun}
\address{Universität Leipzig\\
  Institut für Informatik\\
  PF 100920\\
  04009 Leipzig\\
  Germany}
\email{gabor.braun@informatik.uni-leipzig.de}

\author{Sebastian Pokutta}
\address{Friedrich-Alexander-University of Erlangen-Nürnberg\\
  Department of Mathematics\\
  Cauerstrasse 11\\
91058 Erlangen\\
Germany}
 \email{sebastian.pokutta@math.uni-erlangen.de}

\date{\svnToday/Draft/Revision: \svnInfoRevision}

\begin{abstract}
Extended formulations are an important tool to obtain small (even compact)
formulations of polytopes by representing them as projections of
higher dimensional ones.  It is an important
question whether a polytope admits a \emph{small} extended formulation,
i.e., one involving only a polynomial number of inequalities in its
dimension.
For the case of \emph{symmetric} extended formulations (i.e., preserving
the symmetries of the polytope) Yannakakis established a powerful
technique to derive lower bounds and rule out small formulations.  
We rephrase the technique of Yannakakis in a group-theoretic
framework.  This provides a different perspective on symmetric
extensions and considerably simplifies several lower
bound constructions.
\end{abstract}

\maketitle

\section{Introduction}
\label{sec:introduction}

Extended formulations regained a lot of interest lately (cf., e.g.,
\cite{confExtForm2010}, \cite{faenza2009extended},
\cite{faenza2012}, \cite{combBoundsNNR},\cite{goemans2009cfp},
\cite{kaibel2010symmetry}, \cite{kaibel2011}, \cite{kaibel2011extended},
\cite{pashkovisch2009perm}). The main idea behind extended
formulations is to represent a given
polytope as a projection of a higher dimensional
one, which is usually referred
to as the \emph{extension}.
Whereas at first this may not seem useful,
the higher dimensional polytope might be described
by considerably fewer inequalities.
Hence
it might admit a polynomial time solvable linear program,
if not only the number of inequalities is polynomial,
but, also the coefficients appearing in the projection and the defining
inequalities are appropriately polynomially bounded, e.g., in the dimension. 
Therefore, we are in particular interested in finding
\emph{small} extended formulations, i.e.,
whose size (here measured in the number of inequalities only) is polynomial
in the dimension of the initial polytope.

Due to its appeal
of representing a polytope with an exponential number of inequalities
in polynomial size, in
the 1980s Swart tried to show \(P = NP\)
by devising compact extended formulations for the
traveling salesman problem.  All these formulations shared the commonality
of being \emph{symmetric}, and it was Yannakakis's seminal paper (see
\cite{yannakakis1991expressing}) which put an end to this by showing
that the traveling salesman polytope does not admit a symmetric
extended formulation of polynomial size. In a recent paper
(\cite{fmptw2011}) it was shown that the requirement for symmetry can
be dropped as well and an
unconditional super-polynomial lower bound for the size of any
extended formulation of the traveling salesman polytope was obtained.

At its core Yannakakis's work 
provides techniques for computing the size of an extended formulation via
decomposing slack matrices as the product of two matrices with
non-negative entries.  Moreover, his work establishes a method
for bounding from below the size of symmetric extended formulations.
Using these techniques, he proved, among others, that
the perfect matching
polytope cannot have a symmetric extended formulation of polynomial
size, which was the basis for his impossibility result on the
TSP polytope.

This result was later extended by \cite{kaibel2010symmetry}  
to (weakly-)symmetric extended
formulations of cardinality constrained matching which in contrast do
possess an \emph{asymmetric} extended formulation of polynomial
size.
Similarly, in \cite{goemans2009cfp} an asymmetric extended
formulation of optimal size \(O(n \log n)\) for the permutahedron is
provided, based on AKS-sorting networks.
A symmetric extended formulation for the permutahedron is
the Birkhoff polytope with
\(O(n^2)\) inequalities. This formulation is also optimal in size as established by 
\cite{pashkovisch2009perm}; another example for a gap between the
best symmetric and asymmetric extension.

A more general framework for constructing
(asymmetric) extended formulations by, so called, polyhedral relations
was established in \cite{kaibel2011}. This quite general
method allowed to  recast several constructions of asymmetric
extended formulations (e.g., the \(O(n \log n)\) extended formulation
of the permutahedron) in a unified framework.

\subsection*{Contribution}
We will focus on
\emph{symmetric extended formulations} in this article.
We streamline and extend the lower bound estimation technique of
\cite{yannakakis1991expressing} via algebraic arguments
with the main structure being a group action
expressing the symmetries.

The results of the
algebraic recasting are two compact theorems (Theorem~\ref{th:1}
for general symmetric extended formulations and
Theorem~\ref{thm:superLinLB} for super-linear bounds),
which virtually encapsulate all the necessary
polyhedral and algebraic arguments in black boxes and which provide a
uniform view on symmetric extended formulations. From these black
boxes many known results follow naturally and shortly
(e.g., those in \cite{kaibel2010symmetry}, 
\cite{pashkovisch2009perm}).

We stress that we
do not provide any new or stronger lower bounds but rather
a natural algebraic approach to symmetric extensions as a
different perspective of known results. We
believe that further insights into the
underlying mechanics of Yannakakis's approach can be obtained from
this framework and that the algebraic versions are more amendable to
SDP extensions.
As an indication we formulate Theorem~\ref{thm:1SDP}.
However, we were
unable to derive new lower bounds for SDP extensions. 

As part of streamlining, several technical concepts needed in previous
works could be omitted:
for example an intermediate extension that
has only vertices in \(\binSet\) or indexed families
or partitions compatible with sections.
Moreover, some restrictions were relaxed at no cost:
e.g., the group action can be any affine action and not just
coordinate permutation.

In the process of reformulating the technique we also
obtain several \emph{unnecessary generalizations}, i.e.,
generalizations that do provide further insight into the essence of
the problem but do not lead to stronger lower
bounds.

\subsection*{Outline}
We start with some preliminaries in
Section~\ref{sec:preliminaries} and recall the considered polytopes in
Section~\ref{sec:considered-polytopes}. In
Section~\ref{sec:polytope-pa_n} we study the well-known polytope
\(\pA_n\), which is of special importance in the context of
cutting-planes and whose face lattice is close to that of the parity
polytope.  Then we derive the main
theorem on lower bounds in
Section~\ref{sec:establ-lower-bounds} and reprove Yannakakis's lower
bound for the matching polytope. Next, we conduct a more detailed
analysis of polytopes with small extensions in
Section~\ref{sec:establ-super-line}. We provide significantly
shortened proofs for the lower bounds on the symmetric extension
complexity of the permutahedron and the cardinality indicating
polytope. In Section~\ref{sec:sdp-version-theorem} we provide an SDP
version for one of our main theorems (Theorem~\ref{th:1}).

\section{Preliminaries}
\label{sec:preliminaries}
In the following we briefly recall a few algebraic notions.
As usual, we accompany formal definitions with commutative diagrams to
give a visual representation.
We write maps on the right except for the section map \(s\) for
reasons of readability.
Let \(\log(.)\) denote the logarithm to base \(2\).

\subsection{Symmetric extensions}
\label{sec:symmetric-extensions}

Let \(P \subseteq \R^{m}\) be a polytope.
Recall that an \emph{extension} of \(P\) is
a polytope \(Q \subseteq \R^{d}\) together with
a linear map \(p\colon \R^{d} \to \R^{m}\) satisfying \(Q p = P\). We
use standard notations for group actions as to be found, e.g., in \cite{dixon1996permutation}:
let the group \(G\) act on \(X\) and let \(g \in G\), \(x \in X\) be arbitrary elements.
The action of \(g\) on \(x\) is simply \(g x\); in particular
groups act on the left.

\begin{definition}
 Let \(G\) be a group with an affine group action on
 \(\R^m\). Then \(P \subseteq \R^m\) is a \emph{\(G\)-polytope} if
 \(G\) leaves \(P\) invariant, i.e., \(g P = P\) for all \(g \in
 G\). 
\end{definition}
The group \(G\) will usually be either
the \emph{symmetric group} \(\sg_n\) on \(n\) elements or
the \emph{alternating group} \(\ag_n\) on \(n\) elements.

We will work with symmetric extensions of a \(G\)-polytope
\(P\) defined as follows. 

\begin{definition}
\label{def:symExt}
  A \emph{symmetric extension} of a \(G\)-polytope \(P\)
  is an extension \(Q\) together with \(p \colon Q \to P\) where \(Q\)
  is a \(G\)-polytope 
  and \(p\) is \(G\)-invariant, i.e., \(g (xp) = (gx) p\) for all \(g
  \in G\) and \(x \in Q\). 
\end{definition}

In order to compare extended formulations we define the following
measure.

\begin{definition}
  \label{def:size}
Let \(Q\) be an extension of the polytope \(P\). Then
the \emph{size of \(Q\)} is the number of its facets. The size of the
smallest extension of \(P\) is denoted by \(\xc(P)\) and similarly  the size
of the smallest symmetric extension for a group \(G\)
is denoted by \(\xcs[G](P)\). 
\end{definition}

At first glance Definition~\ref{def:symExt} seems more restrictive than Yannakakis's
one.  However it turns out that Yannakakis's seemingly more
general definition (and also the generalization given in
\cite{kaibel2010symmetry}) does not lead to extended
formulations of smaller size, as we will see at the end of this section.

We further need the notion of a section which assigns to every vertex
in \(P\) a pre-image in \(Q\) under the projection \(p\).

\begin{definition}
  Let \(Q\) and \(P\) be
  \(G\)-polytopes  such that \(Q\) is a symmetric extension of
  \(P\). Then \(s \colon \vertex(P) \rightarrow Q\) is a \emph{section} if
  \(s(x)p = x\) for all \(x \in \vertex(P)\). Further it is an \emph{invariant
  section} if we additionally have \(s(gx) =
g s(x)\)  
for all \(x \in \vertex(P)\) and \(g \in G\). 
\end{definition}

Note that a section \(s\) is usually non-linear. In fact, as pointed
out in \cite{kaibel2010symmetry}, if \(s\) were affine and
\(Q\) an extension of \(P\), then
\(Q \cap \operatorname{aff} \set{s(x)}{x \in X}\) would be
isomorphic to \(P\). Therefore \(Q\) would have
at least as many facets as \(P\),
and so could not have size smaller than \(P\). 

Recall that
a scalar product \(\sprod{.,.}\) on \(\R^n\) is \emph{\(G\)-invariant}
if it is invariant under the linear part of the action of \(G\),
i.e., \(\sprod{gx - g0, gy - g0} = \sprod{x,y}\)
for all \(g \in G\) and \(x,y \in \R^n\).
(The linear part of the \(G\)-action is \(x \mapsto g x - g 0\).)

It is easy to see that there always exist an invariant scalar
product and an invariant section. In fact the invariant section as well as
the invariant scalar product arise from \emph{averaging} over the
group. The proof follows standard arguments; we include it for the sake of
completeness in Appendix~\ref{apx:invScalar}. 

\begin{lemma}
\label{lem:invariantScalarAndSection}
Let \(P \subseteq \R^m\) be a \(G\)-polytope and \(Q \subseteq R^d\) be a
\(G\)-polytope so that \(Q\) is a symmetric extension of
\(P\) with projection \(p\) as
before. Further let \(s: \vertex(P) \rightarrow Q\) be a section and
\(\sprod{.,.}\) be a scalar product on \(\R^d\). Then:
\begin{enumerate}
\item There exists an invariant scalar product \(\ssprod{.,.}\)
  defined (via averaging over the linear part) as
\(\ssprod{x,y} \coloneqq \frac{1}{\size{G}} \sum_{g \in G}
\sprod{g x - g 0, g y - g 0}\),
\item There exists an invariant section \(\bar s\) given by
\[\bar s(x) \coloneqq \frac{1}{\size{G}} \sum_{g \in G}
g^{-1} s(g x).\]
\end{enumerate}
\end{lemma}

The essence of the proof is the celebrated symmetrizing trick.

\subsection{Group actions}
\label{sec:group-actions}

Let \(G\) act on a set \(X\).
Recall, that the \emph{orbit} of an
element \(x \in X\) under \(G\) is defined as \(\orb{G}{x}
\coloneqq \set{\pi x}{\pi \in G}\). The \emph{stabilizer} of an
element \(x \in P\) is the subgroup of elements of \(G\) that leave \(x\)
invariant, i.e., \(\Stab{G}{x} \coloneqq \set{\pi \in G}{\pi x =
  x}\).
Recall the following well-known formula for the size of orbits:

\begin{lemma}[Orbit-Stabilizer Theorem]
\label{lem:stabOrbThm}
Let \(G\) be a finite group acting on a finite set \(X\).
For any \(x \in X\)
we have
 \[\card{\orb{G}{x}} = \card{G:\Stab{G}{x}} = \card{G} / \card{\Stab{G}{x}}.\]
\end{lemma}

In particular, if \(P\) is a \(G\)-polytope, then \(G\) also acts on
the face lattice of \(P\).
We will be interested in the orbits and stabilizers of faces,
for which the following observation and lemma will be helpful.
The observation is just a corollary to Lemma~\ref{lem:stabOrbThm}.

\begin{observation}
\label{obs:facetsToStab}
Let \(P\) be a \(G\)-polytope with \(d\) facets. Then
\[\card{G:\Stab{G}{j}} \leq d\]
for any facet \(j\) of \(P\).
\end{observation}

For a finite set \(Y \subseteq X\), we define \(\ag(Y)\) to be the alternating group
permuting the elements of \(Y\) and leaving \(X \setminus Y\) fixed;
the ambient set will be clear from the context.

\begin{lemma}\citep[Theorem 5.2A]{dixon1996permutation}
\label{lem:smallIndex}
 Let \(G \subseteq \ag_n\) and \(n \geq 10\). Then
\( \card{\ag_n : G} < \binom{n}{k}\) with \( k \leq \frac{n}{2}\)
implies one of the following
\begin{enumerate}
\item there is an invariant subset \(W\) with \(\card{W} < k\) such
  that \(\ag([n]\setminus W)\) is a subgroup of \(G\);
\item  \(\card{\ag_n : G} = \frac{1}{2} \binom{n}{n/2}\) with \(n\) even,
\(\ag_{n/2} \times \ag_{n/2}\) is a subgroup of \(G\), and \(k=n/2\).
\end{enumerate}
\end{lemma}

Note that one can obtain a strengthened version of
Lemma~\ref{lem:smallIndex} by iteratively applying
it to the obtained subgroup.

\subsection{Weakly symmetric extensions}
\label{sec:weakly-symm-extens}

We conclude this section by showing how Yannakakis's concepts fit
into our framework. For this we will use the concept of a
weakly-symmetric extension, which had been used before in
\cite{kaibel2010symmetry}. 
We will show that every weakly-symmetric extension (a generalization
of, both, 
our symmetric extensions and Yannakakis's one) induces
a symmetric one of at most the same size. Therefore weakly-symmetric
extensions do not provide smaller extended formulations and we maintain full
generality by confining ourselves to symmetric extensions while being
able to simplify arguments.

\begin{definition}
A \emph{weakly-symmetric extension} of a \(G\)-polytope \(P\)
is a \(\widetilde{G}\)-polytope \(Q\) together with
a group epimorphism \(\alpha\colon \widetilde{G} \to G\)
and a surjective \(\alpha\)-linear affine map \(p \colon Q \to P\),
i.e., \((\tilde{\pi} x) p =
(\tilde{\pi} \alpha) (x p)\) for all \(\tilde{\pi} \in \widetilde G\)
and \(x \in Q\).

In fact we have the following commutative diagram for all \(\tilde \pi
\in \widetilde{G}\):

\begin{equation*}
\begin{CD}
Q @>\tilde \pi>> Q\\
@VVpV @VVpV\\
P @>\tilde \pi \alpha>> P
\end{CD}    
\end{equation*}
\end{definition}

We now show that weakly-symmetric extended formulations
do not provide smaller formulations than symmetric extended
formulations:

\begin{proposition} 
For every weakly-symmetric extended
formulation \(Q\) of \(P\) with \(Q\subseteq \R^d\) being a \(\tG\)-polytope,
\(P \subseteq \R^m\) being a \(G\)-polytope, projection \(p \colon Q \to P\), and
group epimorphism \(\alpha \colon \tG \to G\),
the restriction to \(R \coloneqq \fixedPoints{\ker \alpha}{Q}\) is
a symmetric extended formulation and \(R\) has dimension and facets at most that of \(Q\).  
\begin{proof}
As \(\ker \alpha\) is a normal subgroup,
\(R = \fixedPoints{\ker \alpha}{Q}\) and
\(X \coloneqq \fixedPoints{\ker \alpha}{(\R^{d})}\)
are invariant under the \(\tG\)-action.
Since the action is affine, \(X\) is an affine subspace.
Thus \(R\) is the intersection of \(Q\)
with the affine subspace \(X\),
and hence it has no higher dimension and no more facets
than \(Q\).

To make \(R\) a \(G\)-polytope,
we define the action of \(g \in G\)
on an element \(x \in R\) via
\[gx \coloneqq \tg x, \quad \tg \alpha = g\]
where \(\tg \in \tG\) is arbitrary so that \(\tg \alpha = g\) holds.
This action is well-defined,
because \(\ker \alpha\) acts trivially on \(R\) by definition,
i.e., whenever \(\tg
\in \ker \alpha\), then \(\tg x = x\) for all \(x \in R\).

It is obvious that the restriction \(p \colon R \to Q\)
preserves the \(G\)-action.
Finally we show that \(R p = P\).
Let \(x \in P\) be arbitrary and choose any
\(y \in Q\) so that \(yp = x\).
As a shorthand notation,
let \(y[H] \coloneqq \frac{1}{\card{H}}\sum_{h \in H} hy\) denote
the group average of \(y\) with respect to any group \(H\).
Then \(y[\ker \alpha] \in R\) and we
have
\[(y[\ker \alpha]) p = (yp) [(\ker \alpha) \alpha] = yp = x,\]
and so the claim follows.
\end{proof}
\end{proposition}

\section{Considered polytopes}
\label{sec:considered-polytopes}
In this section we recall the well-known polytopes that will appear
later.

\subsection{The cardinality indicating polytope}
\label{sec:card-indic-polyt}
The \emph{cardinality indicating polytope} \(P_{card}(n)\) is the
convex hull of all vectors \((x,e_{\onorm{x}})\)
for \(x \in \binSet^n\)
where \(e_{0},\dots,e_{n}\) are linearly independent.
The second vector \(e_{\onorm{x}}\) indicates
the number of \(1\)-entries in \(x\).
\[ P_{card}(n) \coloneqq \conv \set{(x,e_{\onorm{x}})}{x \in \binSet^n}\]
It can be described by the following system
of inequalities (with \( z = \sum_{j=0}^{n} z_{j} e_{j}\)):
\begin{align*}
\sum_{i \in S} x_i & \leq   \sum_{j=0}^{\size{S}} j z_{j} +
\size{S} \sum_{j = \size{S} + 1]}^{n} z_j & \forall\;
\emptyset \nsubseteq S \subseteq [n]\\
\sum_{i \in [n]} x_i &=   \sum_{j = 0}^{n} j z_{j} \\
\sum_{j = 0}^{n} z_j &= 1 \\
 x_{i},z_j &\in [0,1] & \forall\; i \in [n], j = 0, \dots, n
\end{align*}
The cardinality indicating polytope has a symmetric extended
formulation of size \(\Theta(n^2)\) as shown in \cite{koeppe2008intermediate}.

\subsection{The Birkhoff polytope}
The \emph{Birkhoff polytope} \(P_{birk}(n)\) is the convex hull of all doubly
stochastic \(n \times n\) matrices (or equivalently of all \(n \times
n\) permutation matrices). It can be described by the following system
of inequalities:
\begin{align*}
  \sum_{i \in [n]} x_{ij} & = 1  & \forall\;  j \in [n]\\
  \sum_{j \in [n]} x_{ij} & = 1  & \forall\;  i \in [n] \\
  x_{ij} &\in [0,1] & \forall\; i,j \in [n]
\end{align*}

\subsection{The permutahedron}
\label{sec:permutahedron}
The \emph{permutahedron} \(P_{perm}(n)\) is the convex hull of all
permutations of the numbers \(1, \dots, n\), i.e.,
\[P_{perm}(n) \coloneqq \conv \set{\pi(1,\dots,n)}{\pi \in \sg_n}.\]
It can be
described by the following system of inequalities:
  \begin{align*}
\sum_{i \in S} x_i &\geq \frac{\size{S} (\size{S}+1)}{2} & \forall\;
\emptyset \neq S \subseteq [n] \\
\sum_{i \in [n]} x_i &= \frac{n (n+1)}{2}
\end{align*}
and it can be obtained by a projection of the Birkhoff
polytope, i.e., it has a symmetric extended formulation of size
\(O(n^2)\). Also, symmetric extended formulation of the permutahedron
needs at least \(\Omega(n^2)\) inequalities by \cite{pashkovisch2009perm} and so the
Birkhoff polytope is an optimal extension. On the other hand there
exists an asymmetric extended formulation of the permutahedron of size
\(O(n \log n)\) by \cite{goemans2009cfp} which is optimal.

\subsection{The spanning tree polytope}

 For a graph \(G = (V,E)\) and \(U \subseteq V\) let \(E[U]\) denote
the set of edges supported on \(U\). The \emph{spanning
  tree polytope of \(G\) (denoted by: \(P_{STP}(G))\)} is given by the following system of inequalities:
  \begin{align*}
\sum_{e \in E[U]} x_e & \leq \size{U} - 1     &\forall\;  \emptyset
&\neq U \subsetneq V\\
\sum_{e \in E} x_e &= n-1\\
x_e &\in [0,1] & \forall\; e &\in E.
  \end{align*}

There exists an extended formulation of size \(O(n^3)\) due to
\cite{martin1991using} and a lower bound of \(\Omega(n^2)\) follows from
the non-negativity constraints. An interpretation of the associated communication protocol
can be found in \cite{combBoundsNNR}.


\section{The polytope \(\pA_n\)}
\label{sec:polytope-pa_n} 
In the following we consider the well-known polytope \(\pA_n\), which is of
particular interest in the context of cutting-plane procedures. It
realizes maximal rank for all known operators and it represents a
universal obstruction for any admissible cutting-plane procedure (see
\cite{PS20091}). Moreover \(\pA_n\) will serve as an important example
showing that the conditions of Theorem~\ref{thm:orbitCharSuperlinear}
are necessary. The polytope \(\pA_n\) is given by 
\[\pA_n \coloneqq \set{x \in \cube{n}}{\sum_{i \in I} x_i +
  \sum_{i \notin I} (1-x_i) \geq \frac{1}{2} \quad \forall I
  \subseteq [n]}.\]
With \(F_1^n \coloneqq \set{x \in \face{0,1/2,1}^n}{\text{exactly
one entry equal to } 1/2}\) we have \(\pA_n = \conv F_1^n\) (see
e.g., \cite{PS20093}); we drop the index \(n\) if it is clear from the
context. For a vector \(v \in F_1\) let \(\supp_i(v)
\coloneqq \set{j \in [n]}{v_j = i}\). 

We provide a symmetric extended formulation of \(\pA_n\) of size
\(O(n)\).

\begin{theorem}
\label{thm:AnSmallSym}
  Let \(\pA_n\) be defined as above. Then there exists a symmetric
  extended formulation of \(\pA_n\) of size \(O(n)\).
  \begin{proof}
For convenience we translate \(\pA_n\) to \(Q_n:=\pA_n-\frac{1}{2}e\) and
we will provide an extended formulation of \(Q_n\) with \(3n\)
inequalities and \(2n\) variables. Observe that 
\[Q_n \coloneqq \set{x \in
  {\left[ -\frac{1}{2}, \frac{1}{2}\right]}^n}{\size{x_i}
  \leq \frac{1}{2}, \sum_{i \in [n]} \size{x_i} = \frac{n-1}{2} \quad
  \forall i \in
[n]}.\]
While this formulation is polyhedral it is not given by inequalities. However we can
introduce new variables \(y_i\) and \(z_i\) with \(i \in [n]\) and
replace \(\size{x_i}\) with \(y_i + z_i\)
and we obtain a new polytope \(L_n\)
\[L_n \coloneqq \set{(y,z) \in {\left[ 0,\frac{1}{2} \right]}^{2n}}%
{y_i + z_i
  \leq \frac{1}{2}, \sum_{i \in [n]} y_i + z_i = \frac{n-1}{2} \quad
  \forall i \in
[n]}.\]
Observe that \(L_n\) is given by \(3n\) inequalities (\(n\) in the
formulation and \(y_i,z_i \geq 0\) for all \(i \in [n]\)) and \(2n\)
variables. Moreover we claim that with the projection \(p\) defined
via \((y_i,z_i) \mapsto x_{i} = y_i-z_i\) for all \(i \in [n]\) we have
\(p(L_n) = Q_n\). Clearly \(Q_n \subseteq p(L_n)\). For the inverse
inclusion observe that a vertex of \(L_n\) can have only
\(\face{0,1/2}\)-entries. 
  \end{proof}
\end{theorem}

A larger compact extended formulation of size \(O(n^2)\) can be obtained using
Balas's \emph{union of polyhedra} (see \cite{balas85} and
\cite{balas1998dpp}). This formulation only preserves the symmetries
permuting coordinates, however
our extension in Theorem~\ref{thm:AnSmallSym}
preserves the full symmetry group \(Z_2 \wr \sg_n\) of the cube.

We will now derive a lower bound on the extension complexity of
\(\pA_n\). 

\begin{lemma}\cite[Theorem 1]{goemans2009cfp}
\label{lem:genLBextcomp}
Let \(P\)  be any polyhedron in \(\R^n\) with \(v(P)\) vertices. Then
the number of facets \(t(Q)\) of any extended formulation \(Q\) of
\(P\) satisfies 
\[t(Q) \geq \log(v(P)).\]
\end{lemma}

Using Lemma~\ref{lem:genLBextcomp} we obtain the following lower bound
on the extension complexity of \(\pA_n\).

\begin{lemma}
\label{lem:lbAnGeneric}
  Let \(\pA_n\) be defined as above. Then \(\xc(\pA_n) \in
  \Omega(n)\). 
  \begin{proof}
    Observe that \(\size{F_1} = n 2^{n-1}\) and thus by
    Lemma~\ref{lem:genLBextcomp} we obtain \(\xc(\pA_n) \geq \log(n) +
    (n-1) \in \Omega(n)\). 
  \end{proof}
\end{lemma}

Combining Lemma~\ref{lem:lbAnGeneric} and Theorem~\ref{thm:AnSmallSym}
we obtain:

\begin{corollary}
  The symmetric extension complexity \(\xcs(\pA_n) = \xc(\pA_n)\) is \(\Theta(n)\). 
\end{corollary}

One can also obtain an extended formulation of size
\(O(n)\) using reflections at the hyperplanes \(x_i = \frac{1}{2}\)
(see \cite{kaibel2011}), however this formulation is asymmetric.

Finally, we
would like to point out that all results of this section also apply
to the polytope \(\pB_n\) given by
\[\pB_n \coloneqq \set{x \in \cube{n}}{\sum_{i \in I} x_i +
  \sum_{i \notin I} (1-x_i) \geq 1 \quad \forall I
  \subseteq [n]}.\]
This is of particular interest because the parity polytope given by 
\[\text{Par}_n \coloneqq \set{x \in \cube{n}}{\sum_{i \in I} x_i +
  \sum_{i \notin I} (1-x_i) \geq 1 \quad \forall I
  \subseteq [n], |I| \text{ odd}}\]
is closely related to \(\pB_n\) and the cube \(\cube{n}\). In fact,
the face lattice of \(\text{Par}_n\) looks very
much like  \(\pB_n\) or \(\cube{n}\). By the above results
we have \(\xcs(\pB_n),
\xcs(\cube{n}) \in O(n)\), even though \(\xcs(\text{Par}_n) \in \Omega(n
\log n)\) by (\cite{PashPar2011}).

\section{The lower bound black-box for symmetric extended formulations}
\label{sec:establ-lower-bounds}
We will now present the main theorem that we will use in the following
to establish lower bounds. 

\begin{theorem}
  \label{th:1}
  Let a \(G\)-polytope \(Q \subseteq \R^{d}\) be
  a symmetric extension of a \(G\)-polytope \(P \subseteq \R^{m}\).
  For every facet \(j\) of \(Q\)
  let \({\mathcal{F}}_{j}\) be a refinement of
  the \(\Stab{G}{j}\)-orbit partition of
  the vertex set \(X\) of \(P\).
  Then for every real solution to the following inequality system
  in the \(c_{x}\)
  \begin{align}
    \label{eq:19}
    \sum_{x \in X} c_{x} &= 1, \\
    \label{eq:20}
    \sum_{x \in F} c_{x} &\geq 0, & F &\in {\mathcal{F}}_{j},
    \, \text{\(j\) facet of \(Q\)}
  \end{align}
  the point \(\sum_{x \in X} c_{x} x\) lies in \(P\).
\begin{proof}
Let \(\sprod{.,.}\) be an invariant scalar product on \(\R^{d}\).
Let \(n_{j}\) be the normal vector of facet \(j\) pointing inwards.
The inequality of the facet \(j\) is thus of the form
\(\sprod{n_{j}, y} \geq r_{j}\) for some real \(r_{j}\).
These are clearly invariant:
they are permuted together with the facets,
i.e., \(n_{g j} = g n_{j} - g 0\) and \(r_{g j} = r_{j}\)
for all \(g \in G\).

Let \(s \colon X \to Q\) be
an invariant section of \(p\).
Via invariance, the value \(\sprod{n_{j},s(x)} - r_{j}\)
is constant as \(x\) runs through
a \(\Stab{G}{j}\)-orbit.
In particular, it is a constant \(A_{F} \geq 0\) on every
\(F \in {\mathcal{F}}_{j}\);
note that \(F\) is a subset of the vertex set
\(X\) of \(P\).
Thus 
\begin{equation}
  \label{eq:18}
  \sprod{n_{j},\sum_{x \in X} c_{x} s(x)} - r_{j} =
  \sum_{x \in X} c_{x} (\sprod{n_{j},s(x)} - r_{j} ) =
  \sum_{F \in {\mathcal{F}}_{j}} \sum_{x \in F} c_{x} A_{F} \geq 0.
\end{equation}
This shows that \(\sum_{x \in X} c_{x} s(x) \in Q\),
hence applying \(p\) we obtain \(\sum_{x \in X} c_{x} x \in P\).
\end{proof}
\end{theorem}

The result above has a particularly nice
interpretation. When considering a symmetric extension we are allowed
to consider affine combinations of points, rather than convex
combinations, as long as \emph{each sum of coefficients} along an
orbit is non-negative. Put differently, convexity usually requires for
a point to be written as a \emph{convex} combination. In the presence
of symmetry this requirement can be relaxed to an
\emph{affine} combination of points that is convex when averaged over the
orbits.

Theorem~\ref{th:1} can be used to bound the size of extended
formulations as follows. 

\begin{remark}
\label{remark:roadmap}
Suppose we are looking for a
symmetric extended formulation \(Q
\subseteq \R^d\) of a \(G\)-polytope \(P \subseteq \R^m\) with
projection \(p\). Then a
lower bound on the size of \(Q\) (as the number of facets) can be
established in the following way via Theorem~\ref{th:1}:
\begin{enumerate}
\item\label{item:1} Choose a subpartition \(\mathcal
  F_j\) of the
  \(\Stab{G}{j}\)-orbit partition of the vertices of \(P\) for all facets \(j\) of a
  hypothetical \(Q\) of small size.
\item\label{item:2} Find a particular solution \(c_x\) with \(x \in
  X\).
\item\label{item:3} Show that \(\sum_{x \in X} c_x x \notin P\). 
\end{enumerate}  
\end{remark}

Steps~\ref{item:2} and~\ref{item:3} are usually performed
simultaneously by requiring that a solution to the system in
Step~\ref{item:2} violates a valid inequality for \(P\).
This \emph{roadmap} is somewhat similar to Yannakakis's.
However it is more tailored to the
requirements of Theorem~\ref{th:1}. In particular none of the
intermediate steps, such as, e.g., subspace extensions
(defined by equalities and non-negativity constraints) are
needed.

\subsection{Applications to the matching polytope}
\label{sec:matching-polytope}
In this section we will simplify and slightly generalize the result of
\cite{kaibel2010symmetry}, which is itself based on Yannakakis's
technique.
We consider the \(\ell\)-matching polytope of the complete graph
\(K_n = ([n], E_n)\) with \(n \in \N\).
Let \(\mathcal M^\ell(n)\) denote
the set of all matchings of \(K_n\) of size exactly \(\ell\). The
\emph{\(\ell\)-matching polytope} \(\pmatch^\ell(n)\)  is the convex hull of the
characteristic vectors of elements in \(\mathcal M^\ell(n)\), i.e.,
\[\pmatch^\ell(n) \coloneqq \set{\chi(M)}{M \in \mathcal
  M^\ell(n)} \subseteq \cube{E_n}.\]

With \(\sg_n\) acting on the vertices of \(K_n\) by permutation, we have
that \(\pmatch^\ell(n)\) is an \(\sg_n\)-polytope. We will consider
\(\pmatch^\ell(n)\) as an \(\ag_n\)-polytope, i.e., we require less
symmetry for the extension as the \(\ell\)-matching polytope actually possesses. For the
size of any symmetric extended formulation of \(\pmatch^\ell(n)\) we
obtain the following lower bound.

\begin{theorem}
\label{thm:matchingBound}
Let \(n \in \N\) with \(n \geq 10\)  and let \(Q \subseteq \R^d\) be
an \(\ag_{n}\)-symmetric extension of
 \(\pmatch^\ell(n)\). Then the number of facets of \(Q\) is at least 
\[\binom{n}{\lfloor (\ell-1)/2 \rfloor}.\]
\end{theorem}
 The proof is similar to the ones in
\cite{yannakakis1991expressing} and \cite{kaibel2010symmetry} however
we can shorten the argument by using Theorem~\ref{th:1}.
\begin{proof}
First we introduce some notation.
For readability
let \(k \coloneqq \left\lfloor \frac{\ell-1}{2} \right\rfloor\).  

Let \(V\) and \(E\) be the vertex set and
edge set of \(K_n\), respectively.
For a set \(M \subseteq E\),
let \(V(M)\) denote the \emph{support of \(M\)},
i.e., the set of endpoints of all edges in \(M\). Morever,
for \(V_1, V_2 \subseteq V\) and \(M \subseteq E\) let \(M(V_1:V_2)\)
denote the set of edges in \(M\) with one endpoint in \(V_1\) and the
other endpoint in \(V_2\).

Recall that \(\ag_n\) acts on \(V\), \(E\)
and the set of facets of \(Q\).

The proof is by contradiction following the roadmap in
Remark~\ref{remark:roadmap},
so we suppose that \(Q\) has less than \(\binom{n}{k}\)
facets.

Second we define a
subpartition \(\mathcal F_j\) of the
\(\Stab{(\ag_n)}{j}\)-orbit partition of
the vertex set of \(\pmatch^\ell(n)\)
  for all facets \(j\). Let \(j\) be a fixed facet.
  Since the number of
  facets is less than \(\binom{n}{k}\) we have
  \(\card{\ag_n:\Stab{(\ag_n)}{j}} \leq \binom{n}{k}\) by
  Lemma~\ref{obs:facetsToStab}. We apply
  Lemma~\ref{lem:smallIndex} to obtain
  a set \(V_j \subseteq V\) of size at most \(k\)
  for any facet \(j\) of \(Q\)
  so that \(H_j \coloneqq \ag(V \setminus V_j) \subseteq
  \Stab{(\ag_n)}{j}\).
Let us define
for all matching \(W \subseteq E(V_{j} : V)\)
with \(\size{W} \leq \ell\)
\begin{equation}
  \label{eq:1}
  F_{W} \coloneqq \{M \text{ \(\ell\)-matching }\mid M(V_{j} : V) = W \}.
\end{equation}
The family \({\mathcal{F}}_{j}\) is chosen to be the collection
of the non-empty \(F_{W}\), which is easily seen to
refine the orbit partition of
\(H_{j}\) and hence form
a subpartition of
the \(\Stab{(\ag_n)}{j}\)-orbit partition of
\(\vertex(\pmatch^\ell(n))\).

Next we find a solution to the system in
Theorem~\ref{th:1}.  Let \(V_{*}\) and \(V^{*}\)
be arbitrary disjoint subsets of \(V\)
of size \(l_{*}\) and \(l^{*}\), respectively, with
\(l_{*} + l^{*} = 2\ell\).
When \(\ell\) is odd, we select \(l_{*} = l^{*} = \ell\),
and when \(\ell\) is even,
we choose \(l_{*} = \ell - 1\) and \(l^{*} = \ell + 1\).
Thus \(l_{*}\) and \(l^{*}\) are always odd.

Let \(\mathcal M\) denote the set of matchings supported on \(V_{*}
\cup V^{*}\). These matchings are all the vertices of a face of
\(\pmatch^\ell(n)\) (defined by \(x_e = 0\) for all \(e
\notin E(V_{*} \cup V^{*})\)).  Since \(l_{*}\) and \(l^{*}\) are odd,
every such matching must have an odd number of edges
between \(V_{*}\) and \(V^{*}\),
so \(\size{M(V_{*} : V^{*})} \geq 1\) is valid for the face.
We select an affine combination \(\sum_{M \in \mathcal M} c_{M} M\) to
violate this inequality. All other \(c_M\) with \(\ell\)-matching \(M \notin \mathcal
M\) are set to \(0\).
All in all, we need to choose the \(c_{M}\) to satisfy
\begin{align}
  \label{eq:2}
  \sum_{M \in \mathcal M} c_{M} &= 1, \\
  \label{eq:3}
  \sum_{M \in F_{W} \cap \mathcal M} c_{M} &\geq 0, \quad \forall\; W \subseteq E(V_{j} : V)
  \text{ matching}, j \text{ facet of } Q\\
  \label{eq:4}
  \sum_{M \in \mathcal M} c_{M} \size{M(V_{*} : V^{*})} &= 0.
\end{align}
In fact, the chosen \(c_{M}\) will only depend on \(\size{M(V_{*} : V^{*})}\),
so we will set
\begin{equation}
  \label{eq:5}
  b_i = c_{M} \cdot \size{\{M : \size{M(V_{*} : V^{*})} = i\}},
\end{equation}
and let \(\mathcal I\) denote the set of encountered values \(\size{M(V_{*} :
  V^{*})}\). We can simplify the system to
\begin{align}
  \label{eq:6}
  \sum_{i \in \mathcal I} b_{i} &= 1, \\
  \label{eq:7}
  \sum_{i \in \mathcal I} b_{i}
  \frac{\size{\{M \in F_{W} \cap \mathcal M: \size{M(V_{*} : V^{*})} = i\}}}%
  {\size{\{M \in \mathcal M : \size{M(V_{*} : V^{*})} = i\}}}
  &\geq 0,  \quad \forall\; W \text{ as above} \\
  \label{eq:8}
  \sum_{i \in \mathcal I} b_{i} i &= 0.
\end{align}

Now we determine the coefficients in \eqref{eq:7}.
For this we compute the number of matchings \(M\)
with \(\size{M(V_{*} : V^{*})} = i\).
Note that \(\sg(V_{*}) \times \sg(V^{*})\) acts transitively on these
matchings, so the number is the index of
the stabilizer of any such matching
by Lemma~\ref{lem:stabOrbThm}.
The stabilizer consists of the permutations permuting the edges
between \(V_{*}\) and \(V^{*}\), the edges lying completely in \(V_{*}\),
and the edges lying completely in \(V^{*}\).
Also endpoints of
the latter two kinds of edges can be flipped independently,
however not those of the edges between \(V_{*}\) and \(V^{*}\).
So the stabilizer is
\[\Stab{(\sg(V_{*}) \times \sg(V^{*}))}{M} = \sg_{i} \times ({\mathbb{Z}}_{2} \wr \sg_{\frac{l_{*} - i}{2}})
\times ({\mathbb{Z}}_{2} \wr \sg_{\frac{l^{*} - i}{2}}),\]
and its index (by Lemma~\ref{lem:stabOrbThm}) is
\begin{equation}
  \label{eq:9}
  \size{\{M \in \mathcal M : \size{M(V_{*} : V^{*})} = i\}}
  =
  \frac{l_{*}! \cdot l^{*}!}{i!
    \cdot 2^{\frac{l_{*} - i}{2}} \frac{l_{*} - i}{2}!
    \cdot 2^{\frac{l^{*} - i}{2}} \frac{l^{*} - i}{2}!}.
\end{equation}

Next we compute the number of matchings \(M \in F_{W} \cap \mathcal M\)
for which \(\size{M(V_{*} : V^{*})} = i\)
provided that such matchings exist.
Let
\[a_{*} \coloneqq \size{W(V_{*})}, \qquad  a^{*} \coloneqq \size{W(V^{*})}, \qquad a_{*}^{*} \coloneqq \size{W(V_{*} : V^{*})},\]
where \(W(V_{*}) = W(V_{*} : V_{*})\) is
the set of edges in the matching \(W\) supported on
\(V_*\), the set \(W(V^{*})\) is similarly defined,
and \(W(V_{*} : V^{*})\) is the
set of edges with one endpoint in \(V_{*}\) and the other one in \(V^{*}\).
This is essentially the same problem as above with different
parameters. We conclude
\begin{multline}
  \label{eq:13}
  \size{\{M \in F_{W} \cap \mathcal M: \size{M(V_{*} : V^{*})} = i\}}
  \\
  =
  \frac{(l_{*} - 2 a_{*} - a_{*}^{*})!
    \cdot (l^{*} - 2 a^{*} - a_{*}^{*})!}{(i - a_{*}^{*})!
    \cdot 2^{\frac{l_{*} - 2 a_{*} - i}{2}} \frac{l_{*} - 2 a_{*} - i}{2}!
    \cdot 2^{\frac{l^{*} - 2 a^{*} - i}{2}} \frac{l^{*} - 2 a^{*} - i}{2}!}.
\end{multline}
All in all, \eqref{eq:7} expands to
\begin{multline}
  \label{eq:14}
  \sum_{i \in \mathcal{I}} b_{i}
  \frac{2^{a_{*} + a^{*}} \cdot (l_{*} - 2 a_{*} - a_{*}^{*})!
    \cdot (l^{*} - 2 a^{*} - a_{*}^{*})!}{l_{*}! \cdot l^{*}!}
  \cdot
  i (i-1) \dots (i - a_{*}^{*} + 1)
  \\
  \cdot
  \frac{l_{*} - i}{2}
  \left(
    \frac{l_{*} - i}{2} - 1
  \right)
  \dots
  \left(
    \frac{l_{*} - i}{2} - a_{*} + 1
  \right)
  \\
  \cdot
  \frac{l^{*} - i}{2}
  \left(
    \frac{l^{*} - i}{2} - 1
  \right)
  \dots
  \left(
    \frac{l^{*} - i}{2} - a^{*} + 1
  \right)
  \geq 0.
\end{multline}
Observe that this is a polynomial in \(i\) of degree
\(a_{*} + a^{*} + a_{*}^{*} \leq \size{V_{j}} \leq k\)
with a non-negative constant term.
Furthermore \(\size{\mathcal I} \leq k+1\),
as \(\min(l_{*}, l^{*}) = 2k + 1\) and \(\mathcal{I}\) contains only
odd numbers.
Hence to satisfy all the inequalities,
we can choose the \(b_{i}\) such that
\begin{equation}
  \label{eq:15}
  \sum_{i \in \mathcal{I}} b_{i} f(i) = f(0) \quad \deg f \leq k
\end{equation}
for every polynomial \(f\) of degree at most \(k\).

\end{proof}

\section{Establishing quadratic lower bounds}
\label{sec:establ-super-line}
We will now present a technique to establish super linear lower
bounds on the size of symmetric extended formulations. The technique is
 based on \cite{pashkovisch2009perm} however we generalize
previous constructions and provide a uniform, algebraic framework. In fact it
suffices to check few conditions to establish super linear lower
bounds. 

The following theorem will be central to our following discussion. A
similar result had been already established 
in \cite{pashkovisch2009perm} in a combinatorial fashion. We
provide a new, significantly shorter, algebraic proof. 

\begin{theorem}
  \label{thm:orbitCharSuperlinear}
  Let \(Q \subseteq \R^d\) be a symmetric extension
  of an \(\ag_n\)-polytope \(P \subseteq \R^m\).
  Assume that the number  \(N\) of facets of \(Q\)
  is less than \(n(n-1)/2\).
  If \(j\) is a facet of
  \(Q\), then either \(\orb{\ag_n}{j} \cong [n]\) or \(\orb{\ag_n}{j}
  \cong [1]\). In particular, the orbits of the facets of \(Q\) decompose
  \([N]\) into sets of sizes \(n\) and \(1\).
  \begin{proof}
    Let \(j\) be a facet of \(Q\). As \(N < \frac{n(n-1)}{2}\) we obtain \([\ag_n : \Stab{(\ag_n)}{j}]
    < \frac{n(n-1)}{2}\),
    where \(\Stab{(\ag_n)}{j}\) is the stabilizer of \(j\) in \(A_{n}\).
    Applying Lemma~\ref{lem:smallIndex} yields
    that there exists an \(\ag_{n}\)-invariant subset \(W_j\)
    with \(\card{W_j} \leq 1\) such
    that \(\ag([n]\setminus W_j)\) is a subgroup of \(\Stab{(\ag_n)}{j}\).

Since \(\card{W_{j}} \leq 1\),
there does not exist a non-identical permutation of \(W_{j}\), hence
the subgroup \(\ag([n]\setminus W_{j})\) is maximal with the property of
leaving \(W_{j}\) invariant,
so, in fact, \(\Stab{(\ag_{n})}{j} = \ag([n]\setminus W_{j})\).
It follows that
either \(\orb{\ag_{n}}{j} \cong [n]\) (when \(\card{W_{j}} = 1\))
or \(\orb{\ag_{n}}{j} \cong [1]\) (when \(W_{j} = \emptyset\)). This proves the first
part of the claim. The second part follows immediately as the
orbits induce a partition of \([N]\).
\end{proof}
\end{theorem}

Using Theorem~\ref{thm:orbitCharSuperlinear} we will now derive a
sufficient condition for an \(\ag_n\)-polytope to admit only symmetric
extensions of size \(\Omega(n^2)\); in fact the condition can
be applied more widely and \(\binom{n}{2}\) is the limiting case.
The main idea is that a
small symmetric extended formulation has to
average combinatorial properties of the polytope.
The smaller the required size, the more the formulation
averages.
As a consequence, highly asymmetric combinatorial properties
are obstructions to small formulations.
In a slightly more abstract
framework, we can say that the language defined by the vertices of,
say, such a 0/1-polytope is too complex to be decided by a small
symmetric extension.

We would like to stress that the dimension of the polytope
in the next theorem is irrelevant. 

\begin{theorem}
  \label{thm:superLinLB}
  Let \(P\) be an \(\ag_n\)-polytope.
  Let \(J \subseteq [n-1]\) be a non-empty subset of size \(k\).
  For all \(j \in J\),
  let \(H_j \subseteq \ag_{n}\) be a subgroup
  with orbits \(\{1, 2, \dots, j \}\) and \(\{ j+1, \dots, n \}\)
  in \([n]\).
  Then \(\xcs(P) \geq \frac{n k}{2}\) if there exist
  \begin{enumerate}
  \item
    a family \(\set{F_{j}}{j \in J}\) of faces of \(P\)
    such that \(F_j\) is invariant under \(H_j\);
  \item \label{item:4}
    a permutation \(\zeta_j \in \ag_{n}\) for all \(j \in J\) so that
    \(\zeta_j^{-1} ([j]) = [j-1] \cup \face{j+1}\)
    and vertices \(\set{v_{j}}{j \in J}\)
    such that each \(v_{j}\) belongs to \emph{all} the faces \(F_{i}\)
    with \(i \in J\)
    and \(\zeta_j v_{j} \notin F_j\).
  \end{enumerate}
  \begin{remark}
    The above formulation of Theorem~\ref{thm:superLinLB} is tailored
    towards deriving lower bounds:
 for specific polytopes it is particularly easy to
    check the existence of the \(v_{j}\).
    A more theoretical approach is that
    instead of the vertices \(v_{j}\) we require equivalently
    \(\zeta_j F \nsubseteq F_j\)
    where
    \(F \coloneqq \bigcap_{j \in [n-1]} F_j\).
    (In particular,
    \(F \coloneqq \bigcap_{j \in J} F_j \neq \emptyset\)
    is a face.)
    This rephrases the condition completely in the language
    of the face lattice of the polytope.
  \end{remark}
 \begin{proof}[Proof of Theorem~\ref{thm:superLinLB}]
Let \(F \coloneqq \bigcap_{j \in [n-1]} F_j\).
Then \(v_{j} \in F\) and hence
\(\zeta_j F \nsubseteq F_j\) for all \(j \in J\).
In particular, \(F\) is a non-empty face,
so there exists \(v \in \relint(F)\).

First
observe that \(\zeta_j v \notin F_j\) for all
   \(j \in J\): we have \(\zeta_j v \in \relint(\zeta_j F)\),
   and hence \(\zeta_j F\) is the
   smallest face containing \(\zeta_{j} v\).
   Therefore \(\zeta_j v \in F_j\)
   would imply \(\zeta_j F \subseteq F_j\), which contradicts our
   assumption.  We introduce the following notation for symmetrization:
let \(v[G] \coloneqq \frac{1}{\card{G}}\sum_{g \in G} gv\)
the group average of \(v\) with respect to the group \(G\). 

Second we define points \(v_{\epsilon,j}\) for
\(j \in J\) and \(\epsilon > 0\) as follows:
\[v_{\epsilon,j} \coloneqq (1+\epsilon) v [H_j] - \epsilon (\zeta_j
v)[H_j]. \]
Observe that \(v[H_j], (\zeta_j v)[H_j] \in P\). We claim that
\(v_{\epsilon,j} \notin P\) for all \(j \in J\) and \(\epsilon >
0\). As \(F_j\) is \(H_j\)-invariant we obtain that \(v[H_j] \in
F_j\). Similarly, we have that \((\zeta_j v)[H_j] \notin F_j\) as
\(\zeta_j v \notin F_j\). For any \(\epsilon > 0\) the point
\(v_{\epsilon,j}\) lies on the line of \(v[H_j], (\zeta_j v)[H_j]\)
with \(v[H_j]\) separating
\((\zeta_j v)[H_j]\) and \(v_{\epsilon,j}\).
In particular, \(v_{\epsilon,j}\) is on the wrong side of \(F_j\)
(more precisely, it is on the wrong side of any hyperplane
cutting out \(F_{j}\) from \(P\)),
so \(v_{\epsilon,j} \notin P\).
The points \(v_{\epsilon,j}\) will serve as those that any 
symmetric extension of size less than \(n k/2\) fails to cut
off.

Now  let \(Q \subseteq \R^d\) be a symmetric extension of
\(P\), i.e., \(Q\) is itself an \(\ag_n\)-polytope and let \(p\) be the
 associated projection. We choose \(w
 \in Q\) such that \( wp = v\). We define points \(w_{\epsilon,j}\)
 as follows
\[w_{\epsilon,j} \coloneqq  (1+\epsilon) w [H_j] - \epsilon (\zeta_j
w)[H_j].\]
As before we have \(w[H_j], (\zeta_j w)[H_j] \in Q\). Now that \(p\)
is invariant, we obtain that \(w_{\epsilon,j} p =
v_{\epsilon,j}\) for any \(j \in J\) and \(\epsilon > 0\). However,
\(v_{\epsilon,j} \notin P\) and therefore \(w_{\epsilon,j}
\notin Q\) for any \(j \in J\) and \(\epsilon > 0\). We will count how
many facets \(Q\) has to have in
order to ensure this. 

For contradiction, suppose that \(Q\) is given by less than
\(nk/2 \leq n(n-1)/2\) inequalities, hence
Theorem~\ref{thm:orbitCharSuperlinear} applies
and we obtain that the orbits of
facets under \(\ag_n\) are isomorphic either to \([1]\) (fixed
point) or to \([n]\). Let \(T\) be any facet of \(Q\).
If \(w[H_j] \notin T\) then
\(w_{\epsilon,j}\) is on the side of \(T\) pointing inwards
for \(\epsilon\) small enough,
as then \(w_{\epsilon,j}\) is close to \(w[H_j]\). Hence the point
could not be separated and therefore we only have to consider the other case:
\(w[H_{j}] \in T\), i.e.,
for all \(h \in H_j\) we have \(hw \in T\)
and equivalently \(w \in hT\).
Now \(T\) cuts off \(w_{\epsilon,j}\) if and only if
\((\zeta_j w)[H_j] \notin T\).
In other words, there exists \(h \in H_j\) such that
\(w \notin \zeta_j^{-1}hT\).
This is not possible
if the orbit of \(T\) is a fixed point, as it requires both \(w \in
T\) and \(w \notin T\); a contradiction.

If the orbit of \(T\) is isomorphic to \([n]\),
let \(T_{i}\) denote
the face in the orbit corresponding to \(i \in [n]\).
If \(T\) lies in the \(H_{j}\)-orbit \(\{T_{1}, \dots, T_{j}\}\)
then the above conditions state that \(w\) is contained in
\(T_{1}, \dots, T_{j}\) but not in at least one of
\(T_{1}, \dots, T_{j-1}, T_{j+1}\)
(using the condition \(\zeta_j^{-1} ([j]) = [j-1] \cup \face{j+1}\)),
which is only possible
if \(w\) is not contained in \(T_{j+1}\).
Similarly,
if \(T\) lies in the \(H_{j}\)-orbit \(\{T_{j+1}, \dots, T_{n}\}\)
then the above conditions say that \(w\) is contained in
\(T_{j+1}, \dots, T_{n}\) but not in \(T_{j}\).

All in all,
an orbit of facets cuts off \(w_{\epsilon, j}\)
for small \(\epsilon > 0\) if and only if it is isomorphic to \([n]\),
and
\begin{enumerate}
\item
  \(w \in T_{i}\) for all \(i \leq j \phantom{{} + 1}\) 
  but \(w \notin T_{j+1}\), or
\item
  \(w \in T_{i}\) for all \(i \geq j+1\)
  but \(w \notin T_{j}\).
\end{enumerate}
Observe that either case is satisfied by at most one \(j \in [n-1]\)
for a given orbit.
Therefore every orbit can cut off \(w_{\epsilon,j}\)
for small \(\epsilon\) for at most two \(j\).
Hence we need at least \(k/2\) orbits of size \(n\),
so altogether at least \(\frac{n k}{2}\) facets; a contradiction. 
\end{proof}
\end{theorem}

Observe that property~\ref{item:4} from above is very similar to the \emph{basis
  exchange property} of matroids.  In fact the functions \(\zeta_{j}\)
perform such a basis exchange (and possibly more); see Corollary~\ref{cor:matroid}.

\begin{remark}
  Observe that Theorem~\ref{thm:superLinLB} is only
  about a \emph{linear} number of faces of \(P\). It is natural to
  wonder why
  one cannot just \emph{add} these additional constraints. It
  turns out that this is not possible due to the \(\ag_{n}\)-symmetry
  of \(P\). In fact, we would have to add a linear number of cosets of
  facets, each of which is of
  linear size. 
\end{remark}

We shall now provide simplified proofs for known lower bounds using
Theorem~\ref{thm:superLinLB}.  The first two results already appeared in
\cite{pashkovisch2009perm}.

The polytopes we will consider can be found in
\cite{kaibel2010symmetry}, \cite{pashkovisch2009perm}, and
\cite{combBoundsNNR} (see also
Appendix~\ref{sec:considered-polytopes}).

For simplicity,
in the examples we specify explicitly
neither the permutations \(\zeta_{j}\)
nor the groups \(H_{j}\). In fact, the actual choice of \(\zeta_{j}\) does
not matter; a canonical choice is
the transposition \(\zeta_j =(j\ j+1)\).
Moreover, we can always choose
\(H_j \coloneqq
\ag_{n} \cap (\sg_{[j]} \times \sg_{[n] \setminus [j]})\).

\begin{corollary}[Permutahedron]
\label{cor:permLB}
Let \(P_{\textup{perm}}(n) \subseteq \R^n\) be the permutahedron on
\([n]\). Then \(\xcs(P_{\textup{perm}}(n)) \geq \frac{n(n-1)}{2}\).
\begin{proof}

Let \(F_j \coloneqq \face{\sum_{i=1}^{j} x_{i} = \frac{j(j+1)}{2}}\)
for \(j \in [n-1]\) and \(v_{j} = v \coloneqq (1,2, \dots, n)\).
Observe that \(v\) is contained in all the \(F_{j}\)
(in fact, \(\bigcap_{j \in [n-1]} F_j = \face{v}\)).
Clearly, \(F_j\) is invariant under \(H_j\) and we can also
verify that \(\zeta_j v \notin F_j\).
The result now follows from Theorem~\ref{thm:superLinLB}.
\end{proof}
\end{corollary}

With the remark in Section~\ref{sec:permutahedron} this yields
\(\xcs(P_{\textup{perm}}(n)) = \Theta(n^2)\).

\begin{corollary}[Cardinality indicating polytope]
Let \(P_{\textup{card}}(n) \subseteq \R^n\) be the cardinality indicating polytope. Then
\(\xcs(P_{\textup{card}}(n)) \geq \frac{n(n-1)}{2}\).  
  \begin{proof}
Let \[F_j \coloneqq \face{\sum_{i=1}^{j} x_{i} = \sum_{i=1}^{j} i z_{i} +
\sum_{i=j+1}^{n} j z_{i}}\] and choose the \(x\)-part of \(v_j\) to be
\((1,1,\dots,1,0,0,\dots,0) \) with \(1\) appearing \(j\)
times for \(j \in [n-1]\).  We observe that \(v_j
\in F_{i}\) for all \(i\) and, as before,
\(\zeta_j v_j \notin F_j\).
The result follows from Theorem~\ref{thm:superLinLB}.
 \end{proof}
\end{corollary}

Note that the \(\ag_n\)-symmetry of \(P_{\textup{card}}(n)\) permutes only the
entries of \(x\) but leaves the entries of \(z\) unchanged. 
Together with the remark in Section~\ref{sec:card-indic-polyt} we
obtain that \(\xcs(P_{\textup{card}}(n)) = \Theta(n^2)\).

Observe that we can obtain a uniform \(v\),
i.e., \(v \in F\) such that \(\zeta_j v \notin F_j\)
for all \(j \in [n-1]\):
e.g., \(v \coloneqq \frac{1}{n-1}\sum_{j \in [n-1]} v_j\).
In fact, any convex combination of the \(v_j\) (with all
coefficients non-zero) is sufficient. Such an averaged point is not a
vertex however and might be harder to identify right away.  

Often it suffices to 
identify an ascending chain of subsets \(S_1 \subseteq \dots \subseteq
S_{n-1} \subseteq [n-1]\) and derive the \(F_j\) from those. We will
demonstrate this for the case of the spanning tree polytope.

\begin{corollary}[Spanning tree polytope] 
\label{cor:spanTree}
  Let \(P_{\textup{STP}}(K_n)\) be the spanning tree polytope of the complete
  graph \(K_n\) on \(n\) vertices. Then
\(\xcs(P) \geq \frac{n(n-1)}{2}\).  
  \begin{proof}
    Let \(S_j \coloneqq [j]\) and 
\[F_j \coloneqq \face{\sum_{e \in E(S_j)} x_e  = \size{S_j} - 1},\]
where \(E(S_{j})\) denotes the set of edges
between the vertices in \(S_{j}\).
Now let \(v \coloneqq (1,2,\dots,n)\) be the path from
\(1\) to \(n\).  Observe that \(v\) is a vertex of all the \(F_{j}\).
Moreover, we have \(\zeta_j v \notin F_j\) as
\(\zeta_j v\) restricted to \(S_j\) is not a connected graph
and hence does not lie on the facet \(F_j\).
Again we can apply Theorem~\ref{thm:superLinLB}
and the claim follows.
  \end{proof}
\end{corollary}

As mentioned earlier, a lower bound \(\Omega(n^2)\) for the
extension complexity of the spanning tree polytope follows
directly from the non-negativity constraints and Corollary~\ref{cor:spanTree}
highlights that an \(\Omega(n^2)\) lower bound would also follow from
solely examining the remaining constraints; i.e., considering a
different part of the slack matrix.

We will now show that the Birkhoff polytope is an
optimal symmetric extension of itself. This has been also shown in
\cite{combBoundsNNR}, even for non-symmetric extended formulation.
Whereas the
proof for the general case is based on combinatorial rectangle
coverings of the support of the slack matrices, for the symmetric case
the reason for the lower bound is of an algebraic nature and follows
naturally from Theorem~\ref{thm:superLinLB}. 

\begin{corollary}[Birkhoff polytope]
    Let \(P_{\textup{birk}}(n) \subseteq \R^{n^{2}}\) be the Birkhoff
    polytope of \(n \times n\) permutation matrices.
    Let \(A_{n}\) act on \(P_{\textup{birk}}(n)\) via
    permuting the columns of matrices.
    Then
\(\xcs(P) \geq \frac{n(n-1)}{2}\).  
\begin{proof}
Let \(F_j \coloneqq \face{\sum_{i=1}^{j} x_{j+1,i} = 0}\),
which is the intersection of \(x_{j+1,i} \geq 0\) for \(i \in [j]\).
Then \(\bigcap_{j=1}^{n-1} F_{j}\) is just
the vertex \(v\) with \(v_{i,i} = 1\) for all \(i\).
It is easy to see that \(\zeta_j v \notin F_j\) and clearly
\(F_j\) is invariant under \(H_j\). The result follows with
Theorem~\ref{thm:superLinLB}.
\end{proof}
\end{corollary}

We will now provide an example showing that the conditions specified in Theorem~\ref{thm:superLinLB} are
necessary. In particular we show why Theorem~\ref{thm:superLinLB}
fails for \(k \geq 5\) when applied to \(\cube{n}\); the standard
formulation of the cube has \(2n\) inequalities and \(k \geq 5\) would
imply a lower bound of \(\frac{5}{2}n > 2n\). In fact
Theorem~\ref{thm:superLinLB} fails already for \(k \geq 3\). 

\begin{example}[Applying Theorem~\ref{thm:superLinLB} to \(\cube{n}\)]
Contrary to intuition,
the cube \(\cube{n}\) has only small families \(J\) of faces
satisfying the condition of Theorem~\ref{thm:superLinLB}.
In particular, all the families contain \emph{at most two} faces.
We are now providing a direct proof.
 
The proper faces \(F_{j}\) with stabilizer orbits
\(\{1, 2, \dots, j \}\) and \(\{ j+1, \dots, n \}\)
are only
\begin{align*}
  x_{1} = x_{2} = \dots = x_{j} &= 0, \\
  x_{1} = x_{2} = \dots = x_{j} &= 1, \\
  x_{j+1} = \dots = x_{n} &= 0, \\
  x_{j+1} = \dots = x_{n} &= 1.
\end{align*}
Note that the family of faces cannot include, e.g.,
\(F_{j} = \face{x_{1} = \dots = x_{j} = 0}\) and
\(F_{k} = \face{x_{1} = \dots = x_{k} = 0}\)
for \(j < k\).
Otherwise
\begin{gather}
  \zeta_{j} v_{j} \in \zeta_{j} F_{k} = \face{x_{\zeta_{j}(1)}=
    \dots = x_{\zeta_{j}(k)} = 0} \subseteq F_{j},
  \intertext{as}
  [j] = \zeta_{j}([j-1] \cup \face{j+1}) \subseteq \zeta_{j}([k]).
\end{gather}
Moreover, as \(\face{x_{1} = \dots = x_{j} = 0}\) and
\(\face{x_{1} = \dots = x_{k} = 1}\) are disjoint,
they cannot be both contained in the family.

Therefore the family can contain at most one of the faces
of the form
\(\face{x_{1} = \dots = x_{j} = 0}\) and
\(\face{x_{1} = \dots = x_{j} = 1}\).
Similarly, it contains at most
one of the other faces:
\(\face{x_{j+1} = \dots = x_{n} = 0}\) and
\(\face{x_{j+1} = \dots = x_{n} = 1}\).
This implies a total of \(2\) faces at most.
\end{example}

We conclude this section with a matroid version of 
Theorem~\ref{thm:superLinLB}.
In this case Condition~\ref{item:4}
asks for (repeated) failure of the basis-exchange property.

A matroid \(\mathcal M = (E, \mathcal F)\)
is a \emph{\(G\)-matroid} for some group \(G\),
if \(G\) acts on \(E\) preserving the independent sets,
i.e., \(\pi F \in
\mathcal F\) for all \(\pi \in G\) and \(F \in \mathcal F\). 

\begin{corollary}
  \label{cor:matroid}
  Let \(\mathcal M = (E, \mathcal F)\) be an \(\ag_{n}\)-matroid with rank function
  \(r\).
  Furthermore,
  let \(J \subseteq [n-1]\) be a non-empty subset of size \(k\).
  For all \(j \in J\),
  let \(H_j \subseteq \ag_{n}\) be a subgroup
  with orbits \(\{1, 2, \dots, j \}\) and \(\{ j+1, \dots, n \}\).
  Let \(P \coloneqq \set{ x \in
    \cube{E}}{\sum_{e \in F} x_e \leq r(F)}\) be the
  independent set polytope associated with \(\mathcal M\). 
  Then \(\xcs(P) \geq \frac{n k}{2}\) if there exist
  \begin{enumerate}
  \item
    a family \(\set{F_{j}}{j \in J}\) of flats of
    \(\mathcal M\)
    such that \(F_j\) is invariant under \(H_j\);
 \item\label{item:5}
    a permutation \(\zeta_j \in \ag_{n}\) and \(S_j \in \mathcal F\)
    for all \(j \in J\) so that
    \(\zeta_j^{-1} [j] = [j-1] \cup \face{j+1}\) and \(\size{S_j \cap
      F_i} = r(F_i)\) for all \(i \in J\),
    but \(\size{\zeta_j S_j \cap F_j} < r(F_j)\).
 \end{enumerate}
 \begin{proof}
   Follows immediately from Theorem~\ref{thm:superLinLB} with faces \(
  \face{\sum_{e \in F_j} x_e = r(F_j)}\) for \(j \in J\).
 \end{proof}
\end{corollary}

\section{SDP-version of Theorem~\ref{th:1}}
\label{sec:sdp-version-theorem}
In Section~\ref{sec:establ-lower-bounds} we established the key result
for bounding the size of symmetric extended formulations where the
extension is a polytope. We will now extend Theorem~\ref{th:1} to the
case where the extension is a semidefinite program (SDP). 

Given
two square matrices \(A, B \in \R^{n \times n}\) with \(n \in \N\),
the (standard) Frobenius (inner-) product of \(A\) and \(B\) is defined as 
\[A \bullet B \coloneqq \sum_{i,j \in [n]} A_{ij} B_{ij}.\] 
If a square matrix \(A \in \R^{n \times n}\) is positive semidefinite,
we write \(A \succeq 0\) as usual.
An \emph{SDP} is an optimization problem of
\begin{align*}
\min\  C \bullet X & \\
s.t.\ A_j \bullet X &= b_j \qquad j \in [f]\\
X & \succeq 0,
\end{align*}
where \(f \in \N\) and \(A_j, C, X \in R^{m \times m}\) are
symmetric square
matrices with \(j \in [f]\).
Slightly abusing
notions we will use the term SDP to refer to the feasible region of an
SDP; we are not interested in any particular objective function. Given
a group \(G\), a feasible region of an SDP \(Q\) is a \emph{\(G\)-SDP}
if \(gQ = Q\) and
\(gX \succeq 0\) whenever \(X \succeq 0\) for all \(g \in G\). Note
that the second requirement ensures that the action of \(G\) preserves
the positive semidefinite cone. In a
first step we will establish the existence of a \(G\)-invariant
Frobenius product, i.e., for \(A,B \in \R^{m \times m}\) we have
\(A \bullet B = gA \bullet gB\). The following lemma is the analog of
Lemma~\ref{lem:invariantScalarAndSection}. 

\begin{lemma}
\label{lem:invarFrob}
  Let \(G\) be a group acting linearly and faithfully on \(\R^{m \times m}\). Then
  there exists a \(G\)-invariant Frobenius product
  defined as
\[A \bar{\bullet} B \coloneqq \frac{1}{\size{G}} \sum_{g \in G} gA
  \bullet gB
\]
with \(A,B \in \R^{m \times m}\).
  \begin{proof}
Let \(\pi \in G\) and \(A,B \in \R^{m \times m}\). As before we have
\[\pi A \bar \bullet \pi B = \frac{1}{\size{G}} \sum_{g \in G} \pi gA
  \bullet \pi gB = \frac{1}{\size{G}} \sum_{g \in G} gA
  \bullet gB  = A \bar \bullet B.\]
  \end{proof}
\end{lemma}

\begin{definition}
A \emph{symmetric SDP-extension} of a \(G\)-polytope \(P\)
is a \(\widetilde{G}\)-SDP \(Q\) together with
a group epimorphism \(\alpha\colon \widetilde{G} \to G\)
and linear map \(p \colon \R^{d \times d} \to \R^{m}\) that is also
\(\alpha\)-linear, i.e., \(p\) has to satisfy \(Q p = P\) and \((\tilde{\pi} Q) p =
(\tilde{\pi} \alpha) (Q p)\) for all \(\tilde \pi \in \widetilde G\). 
\end{definition}

We are ready to prove the SDP-variant of Theorem~\ref{th:1}.

\begin{theorem}
\label{thm:1SDP}
  Let a \(\widetilde{G}\)-SDP \(Q \subseteq \R^{d \times d}\) be
  a symmetric SDP-extension of a \(G\)-polytope \(P \subseteq \R^{m}\)
  via \(\alpha\colon \widetilde{G} \to G\)
  and an \(\alpha\)-linear map \(p \colon \R^{d \times d} \to \R^{m}\).
  For every facet \(j\) of \(Q\)
  let \({\mathcal{F}}_{j}\) be a refinement of
  the \(\Stab{\widetilde{G}}{j}\)-orbit partition of
  the vertex set \(V\) of \(P\) and  let \(s \colon V \rightarrow Q\)
  be a section.
  Then for every real solution to the following inequality system
  in the \(c_{v}\)
  \begin{align}
    \label{eq:16}
    \sum_{v \in V} c_{v} &= 1, \\
    \label{eq:17}
    \sum_{v \in F} c_{v} s(v) &\succeq 0, & F &\in {\mathcal{F}}_{j},
    \, \text{\(j\) facet of \(Q\)}
  \end{align}
  the point \(\sum_{v \in V} c_{v} v\) lies in \(P\).
  \begin{proof}
Let \(\bullet\) be a \(\widetilde G\)-invariant Frobenius product on
\(\R^{d\times d}\) and let \(Q\) be given with respect to that
product in the form
\[Q = \set{X \in \R^{d \times d}}{A_j \bullet X = b_j\ \forall j
  \in [f], X \succeq 0},\]
with \(f \in \N\) and \(A_j \in \R^{d \times d}\) symmetric for all
\(j \in [f]\).
Obviously,
\begin{equation}
\label{eq:21}
A_{j} \bullet (\sum_{v \in V} c_v s(v)) - b_{j} =
\sum_{v \in V} c_v (A_{j} \bullet s(v) - b_{j} ) =
0.
\end{equation}
Moreover we have that 
\[\sum_{v \in V} c_v s(v) =
\sum_{F \in {\mathcal{F}}_{j}} \sum_{v \in F} c_{v} s(v) \succeq 0.\]
This shows that \(\sum_{v \in V} c_{v} s(v) \in Q\),
hence applying \(p\) we obtain \(\sum_{v \in V} c_{v} v \in P\).
  \end{proof}
\end{theorem}

\section*{Acknowledgements}
The authors would like to thank Samuel Fiorini, Volker Kaibel,
Kanstantsin Pashkovich, and Hans R.~Tiwary for the helpful
discussions and the several insights that improved our work.  

\bibliographystyle{plainnat}
\bibliography{bibliography}

\appendix

\section{Invariant scalar products and sections}
\label{apx:invScalar}

\begin{ScalarLemma}
Let \(P \subseteq \R^m\) be a \(G\)-polytope and \(Q \subseteq R^d\) be a
\(G\)-polytope so that \(Q\) is a symmetric extension of
\(P\) with projection \(p\) as
before. Further let \(s: \vertex(P) \rightarrow Q\) be a section and
\(\sprod{.,.}\) be a scalar product on \(\R^d\). Then:
\begin{enumerate}
\item There exists an invariant scalar product \(\ssprod{.,.}\) defined as 
\[\ssprod{x,y} \coloneqq \frac{1}{\size{G}} \sum_{g \in G}
\sprod{g x - g 0, g y - g 0},\]
\item There exists an invariant section \(\bar s\) given by
\[\bar s(x) \coloneqq \frac{1}{\size{G}} \sum_{g \in G}
g^{-1} s((g \alpha) x).\]
\end{enumerate}
\begin{proof}
To simplify calculations for the scalar product,
we confine ourselves to linear group actions as it suffices to
consider the linear part of an action. We therefore assume that \(g 0 = 0\)
for \(g \in G\); note that we can do this without 
loss of generality. 
  Let \(\ssprod{.,.}\) be defined as above. We claim that
  \(\ssprod{.,.}\) is a well-defined scalar product such that
  \[\ssprod{g x, g y} = \ssprod{x, y} \]
  for all \(x,y \in \R^d\) and \(g \in G\). Observe that
  \(\ssprod{.,.}\) is a symmetric bilinear function.
  Moreover, \(\ssprod{x,x} = \frac{1}{\size{G}} \sum_{g \in G, i \in [n]}
  \sprod{g x, g x} > 0\) for \(x \neq 0\).
  Therefore \(\ssprod{.,.}\) is a
well-defined scalar product.
In order to show that it is invariant
under the action of \(G\), let \(\pi \in G\) and observe 
\[\frac{1}{\size{G}} \sum_{g \in G}
\sprod{g x, g y} =  \frac{1}{\size{G}} \sum_{g \in G}
\sprod{g \pi x, g \pi y} = \ssprod{\pi x, \pi y},\]
as \(g \pi\) runs through \(G\), when \(g\) does so,
because \(G\) is a group.

Now consider \(\bar s(x)\), let \(x \in \vertex(P)\), and let
\(\pi \in G\). The map \(\bar s(x)\) is indeed a
section, as
\[\bar s (x) p =  \frac{1}{\size{G}} \sum_{g \in G}
  g^{-1} s(g x) p = \frac{1}{\size{G}} \sum_{g \in G}
  g^{-1} (gx)
  = x.\]
For \(\pi \in G\) we have 
\begin{align*}
  \pi \bar s(x) = &\ \pi \left( \frac{1}{\size{G}} \sum_{g \in G}
g^{-1} s(g x) \right)  = \frac{1}{\size{G}} \sum_{g \in G}
\pi g^{-1} s(g x) \\
= &\ \frac{1}{\size{G}} \sum_{g \in G}
g^{-1} s(g \pi x) = \bar s(\pi x)
\end{align*}
\end{proof}
\end{ScalarLemma}
\end{document}